\documentclass[11pt,reqno, final]{amsart}

\usepackage{color}
\usepackage[colorlinks=false]{hyperref}
\usepackage{amsmath, amssymb, amsthm}
\usepackage{mathtools}
\usepackage[noabbrev,capitalize,nameinlink]{cleveref}
\crefname{equation}{}{}
\usepackage{fullpage}
\usepackage{graphics}
\usepackage{pifont}
\usepackage{tikz}
\usepackage{bbm}
\usepackage[T1]{fontenc}

\usetikzlibrary{arrows.meta}

\usepackage{environ}
\usepackage{framed}
\usepackage{url}
\usepackage[linesnumbered,ruled,vlined]{algorithm2e}
\usepackage[noend]{algpseudocode}
\usepackage[labelfont=bf]{caption}
\usepackage{cite}
\usepackage{framed}
\usepackage[framemethod=tikz]{mdframed}
\usepackage{appendix}
\usepackage{graphicx}
\usepackage[textsize=tiny]{todonotes}
\usepackage{tcolorbox}
\allowdisplaybreaks[1]

\crefname{algocf}{Algorithm}{Algorithms}

\crefname{equation}{}{} 
\AtBeginEnvironment{appendices}{\crefalias{section}{appendix}} 

\usepackage[color]{showkeys} 

\colorlet{refkey}{orange!20}
\colorlet{labelkey}{blue!30}

\crefname{algocf}{Algorithm}{Algorithms}

\numberwithin{equation}{section}
\newtheorem{theorem}{Theorem}[section]

\newtheorem{lemma}[theorem]{Lemma}

\crefname{claim}{Claim}{Claims}

\newtheorem{corollary}[theorem]{Corollary}

\newtheorem*{question*}{Question}

\theoremstyle{definition}
\newtheorem{definition}[theorem]{Definition}

\newtheorem*{definition*}{Definition}

\theoremstyle{remark}
\newtheorem*{remark}{Remark}

\newcommand{\norm}[1]{\left\lVert#1\right\rVert}
\newcommand{\snorm}[1]{\lVert#1\rVert}

\newcommand{\mb}{\mathbb}

\newcommand{\mc}{\mathcal}
\newcommand{\mr}{\mathrm}

\newcommand{\Balance}{\textsc{Balance}}
\renewcommand{\l}{\langle}
\renewcommand{\r}{\rangle}
\newcommand{\eps}{\varepsilon}
\newcommand{\assign}{\leftarrow}
\newcommand{\pe}{\preceq}
\newcommand{\se}{\succeq}
\renewcommand{\frak}{\textbf}
\newcommand{\komlos}{Koml\'{o}s}
\newcommand{\disc}{\mr{disc}}

\allowdisplaybreaks

\title{Discrepancy Minimization via a Self-Balancing Walk}

\author[A1]{Ryan Alweiss}
\address{Department of Mathematics, Princeton University, Princeton, NJ, 08544, USA}

\author[A2]{Yang P. Liu}
\address{Department of Mathematics, Stanford University,
Stanford, CA 94305, USA}

\author[A3]{Mehtaab Sawhney}
\address{Department of Mathematics, Massachusetts Institute of Technology, Cambridge, MA 02139, USA}
\email{alweiss@princeton.edu, yangpliu@stanford.edu, msawhney@mit.edu}

\begin{document}
\begin{abstract}
We study discrepancy minimization for vectors in $\mb{R}^n$ under various settings. The main result is the analysis of a new simple random process in multiple dimensions through a comparison argument. As corollaries, we obtain bounds which are tight up to logarithmic factors for several problems in online vector balancing posed by Bansal, Jiang, Singla, and Sinha (STOC 2020), as well as linear time algorithms for logarithmic bounds for the Koml\'{o}s~conjecture.
\end{abstract}

\maketitle

\section{Introduction}
We start with discussing the \emph{vector balancing problem} -- given vectors $v_1, v_2, \cdots, v_t \in \mb{R}^n$, pick signs $\eps_1, \eps_2, \cdots, \eps_t \in \{-1, 1\}$ so that the \emph{discrepancy} $\|\sum_{i=1}^t \eps_iv_i\|_\infty$ is as small as possible. This problem encompasses many known problems in discrepancy theory, including the \komlos~conjecture and minimizing set discrepancy. Concretely, the \komlos~conjecture asks for the best bound $B(n,t)$ such that for any matrix $A \in \mb{R}^{n \times t}$ with columns with $\ell_2$-norm norm at most $1$, there is a signing $\eps \in \{-1,1\}^t$ with $\|A\eps\|_\infty \le B(n,t)$. The \komlos~conjecture states that one may take $B(n,t) = O(1)$. The best known bound is $B(n,t)=O(\sqrt{\log \min(n,t)})$ due to Banaszczyk \cite{Ban98,Ban12}. Similarly, the problem of minimizing set discrepancy considers the case where $v_1, v_2, \cdots, v_t \in \{0,1\}^n$. Here, Spencer's famous ``six standard deviations suffice'' \cite{Spe85} shows that there exists a signing $\eps_1, \eps_2, \cdots, \eps_t \in \{-1, 1\}$ so that $\|\sum_{i=1}^t \eps_iv_i\|_\infty = O(\sqrt{t\max(1,\log(t/n))})$. In particular, if $t = n$, then discrepancy $6\sqrt{n}$ is achievable, and this is tight up to the constant $6$.

While the original proofs of Banaszczyk and Spencer \cite{Ban98,Ban12,Spe85} were nonconstructive, there has been significant interest in finding algorithmic versions.
The first major results in this direction, which gave polynomial time algorithms for Spencer's result \cite{Spe85}, were achieved by Bansal \cite{Bansal10} and Lovett-Meka \cite{LM15}. Since then, there have been several other constructive discrepancy minimization algorithms \cite{Rothvoss17,ES18,BDG16,BG17,BDGL18,DNTT18}, including those matching Banaszczyk's bound for the \komlos~conjecture \cite{Ban98,Ban12} due to Bansal, Dadush, Garg \cite{BDG16} and Bansal, Dadush, Garg, Lovett \cite{BDGL18}. However, these algorithms do not currently seem to extend to an online setting.

\subsection{Online algorithms.}
The problem of \emph{online discrepancy minimization}, proposed by Spencer \cite{Spencer77}, is to assign weights $\eps_i \in \{-1, 1\}$ to vectors $v_1, v_2, \dots, v_t \in [-1,1]^n$ which arrive one at a time, while trying to maintain a low $\ell_\infty$ norm of all the partial sums $w_i = \sum_{j=1}^i \eps_jv_j.$ Against adaptive adversaries, the best possible bound is $\Omega(\sqrt{t})$ as the adversary may choose the next vector $v_{i+1}$ to be orthogonal to $w_i$. Furthermore, a random signing achieves a bound of $O(\sqrt{t \log n})$.

However, it was open whether one could get any improvement over the trivial $O(\sqrt{t \log n})$ bound in the \emph{oblivious} version of this vector balancing problem. Here, an adversary fixes vectors $v_1,\ldots,v_t \in [-1,1]^n$ ahead of time, and the player may use randomness. No deterministic algorithm can do better than $\Omega(\sqrt{t})$ as this is essentially equivalent to the case of an adaptive adversary.  Our main result provides a probabilistic algorithm that achieves an $\ell_\infty$ bound of $O(\max_{i\in[t]} \|v_i\|_2 \log(nt/\delta))$ with probability at least $1-\delta$ for any $\delta > 0$. For $v_i$ in $[-1,1]^{n}$, our result, assuming no further structure on the $v_i$, gives a bound of $O(\sqrt{n} \log(nt/\delta)).$

The best previous results in this direction instead required that $v_i$ were sampled in an iid manner from a distribution $\frak{p}$ which is known beforehand. In this direction there are a number of previous works each achieving different garuantees. In the most restrictive setting where the vectors $v_i$ were sampled uniformly from $[-1,1]^n$, Bansal and Spencer \cite{BS19} showed an $\ell_{\infty}$ guarantee of $O(\sqrt{n} \log t)$. In the more general setting where when the vectors $v_i$ were sampled from $\frak{p}$ supported on $[-1,1]^n$, Aru, Narayanan, Scott, and Venkatesan \cite{ANSV18} achieved a bound of $O_n(\sqrt{\log t})$ (where the implicit dependence on $n$ is super-exponential) and Bansal, Jiang, Singla, and Sinha \cite{BJSS19} achieved a $\ell_{\infty}$ guarantee of $O(n^2\log(nt))$.

\subsection{Algorithm motivation and description.}
We now motivate and describe the self-balancing walk which is at the heart of the paper. Let the input vectors be $v_1, \cdots, v_t$, where $\|v_i\|_2 \le 1$ for all $i \in [t]$, and let $\eps_1, \cdots, \eps_t \in \{-1,1\}$ be chosen later. The algorithm maintains the current partial sum $w_{i-1} = \eps_1v_1 + \cdots + \eps_{i-1}v_{i-1}$, and will decide the sign of $v_i$ probabilistically depending on $w_{i-1}$ and $v_i$. First, it is natural that the algorithm should pick $\eps_i = -1, 1$ with probability $1/2$ each if $w_{i-1}$ and $v_i$ are orthogonal vectors. Additionally, the more correlated $v_i$ and $w_{i-1}$ are, the higher probability that the algorithm picks $\eps_i = -1.$

A natural choice is for the algorithm to pick $\eps_i = 1$ with probability $\frac12 - \frac{\l w_{i-1},v_i\r}{2c},$ where $c = 30\log(nt/\delta)$ is a constant upper bound on $|\l w_{i-1},v_i\r|$ with high probability. This way, if $w_{i-1}$ and $v_i$ are orthogonal, the algorithm picks $\eps_i = -1, 1$ equiprobably. Additionally, the bias is stronger further from the origin, and it is stronger if one of the signings reduces the norm by a larger amount. The fact that the probability $\frac12 - \frac{\l w_{i-1},v_i\r}{2c}$ is linear in $w_{i-1}$ is important for our analysis, hence why we must choose this parameter $c$ in the algorithm.

\begin{algorithm}[h]
\caption{$\Balance(v_1,\cdots,v_t,\delta)$ -- takes a sequence of input vectors $v_1,\cdots,v_t$ and assigns them $\pm 1$ signs online to maintain low discrepancy with probability $1-\delta$.}
$w_0 \assign 0.$ \\
$c \assign 30 \log(nt/\delta).$ \\
\For{$1 \le i \le t$}{
    \If{$|\l w_{i-1},v_i \r| > c$ \emph{or} $\|w_{i-1}\|_\infty > c$ }{
        \textbf{Fail}. \Comment{Algorithm terminates with failure.} \label{line:fail} \\
    }
    $p_i \assign \frac{1}{2} - \frac{\l w_{i-1}, v_i\r}{2c}.$ \label{line:prob-init}\\
    $\eps_i \assign 1$ with probability $p_i$, and $\eps_i \assign -1$ with probability $1-p_i.$ \label{line:prob} \\
    $w_i \assign w_{i-1} + \eps_iv_i.$ \label{line:move}
}
\label{algo:balance}
\end{algorithm}

Our main result is that \cref{algo:balance} maintains low discrepancy with high probability.
\begin{theorem}
\label{thm:balance}
For any vectors $v_1, v_2, \cdots, v_t \in \mb{R}^n$ with $\|v_i\|_2 \le 1$ for all $i \in [t]$, algorithm $\Balance(v_1, \cdots, v_t, \delta)$ maintains $\|w_i\|_\infty = O\left(\log(nt/\delta)\right)$ for all $i \in [t]$ with probability $1-\delta$.
\end{theorem}

We also note that this theorem is sharp up to logarithmic factors in $n$ and $t$ due to a lower bound of $\Omega(\sqrt{\log t/\log \log t})$ given in \cite{BJSS19}. It seems possible to the authors that a variant of \cref{algo:balance} can maintain an $\ell_\infty$ bound of $O(\sqrt{\log nt})$ instead of $O(\log nt)$. This is an interesting open problem.

\subsection{Consequences of \cref{thm:balance}}
\cref{algo:balance} works against oblivious adversaries. Therefore, \cref{thm:balance} implies tight bounds up to logarithmic factors in $n$ and $t$ for all of Questions 1-5 in \cite{BJSS19}. We state Questions 4 and 5, which are about oblivious adversaries, as these generalize the stochastic and prophet models discussed in the other questions raised in \cite{BJSS19}.
\begin{itemize}
    \item  \cite[\S8, Question~4]{BJSS19} Is there an online algorithm which maintains discrepancy $\mr{poly}(n, \log t)$ on any sequence of vectors in $[-1,1]^n$ chosen by an oblivious adversary?
    \item \cite[\S8, Question~5]{BJSS19} Is there an online algorithm which maintains discrepancy $\mr{poly}(s, \log n,$ $\log t)$ on any sequence of $s$-sparse vectors in $[-1,1]^n$ chosen by an oblivious adversary?
\end{itemize}
In fact, \cref{thm:balance} directly implies a nearly tight bound of $O(\sqrt{n} \log(nt))$ for Question 4, and $O(\sqrt{s} \log(nt))$ for Question 5.

We also get improved bounds to the online geometric discrepancy problems of online interval discrepancy and online Tusn\'{a}dy's problem by using a simplified version of the reduction to vector balancing given in \cite{BJSS19}. This is discussed in \cref{sec:applications}.

Finally we also obtain linear time algorithms for logarithmic bounds for the \komlos~conjecture. In what follows, $\mr{nnz}(A)$ denotes the number of nonzero entries in the matrix $A$.
\begin{theorem}
\label{thm:komlos}
Given a matrix $A \in \mb{R}^{n \times t}$ with columns with $\ell_2$-norm at most $1$, we can find with high probability in $O(\mr{nnz}(A))$ time a vector $x \in \{-1,1\}^t$ such that $\|Ax\|_\infty = O(\sqrt{\log t \cdot \log n})$.
\end{theorem}
This requires a minor modification of \cref{algo:balance} which we sketch at the end of \cref{sec:analysis}. Previous constructive discrepancy minimization algorithms \cite{Bansal10,LM15,BDG16,ES18,BDGL18} involved expensive linear algebra or solving linear or semidefinite programs, although \cite{BDG16,BDGL18} achieve stronger bounds of $\|Ax\|_\infty = O(\sqrt{\log t}), O(\sqrt{\log n})$ respectively. This result therefore can be seen as a stepping stone towards giving input-sparsity time algorithms for discrepancy problems, a direction mentioned by Dadush \cite{DadWeb}.

\subsection{Previous approaches for algorithmic discrepancy minimization.} Here we describe previous approaches to algorithm discrepancy minimization and the difficulties in extending previous methods to the online setting against oblivious adversaries. Previous approaches \cite{Bansal10,LM15,ES18,BDG16,BDGL18} either solve linear or semidefinite programs or perform random walks on the sign vector $(\eps_1, \cdots, \eps_t)$, all which require knowing all input vectors $v_1, \cdots, v_t$ at the beginning.

The results of \cite{BS19} and \cite{BJSS19} work for the restricted online setting where all vectors come from a fixed distribution $\frak{p}$ and work by choosing the sign $\eps_i = -1, 1$ to minimize a potential function of the current point such as $\Phi(w) = \sum_{i=1}^n \cosh(\lambda w_i)$. This approach has significant difficulties working in the setting of oblivious adversaries as algorithms minimizing potentials are deterministic, and a lower bound of $\Omega(t^{1/2})$ holds against any deterministic algorithm against oblivious adversaries.

\subsection{Overview of analysis of \cref{algo:balance}}
A natural approach to analyzing \cref{algo:balance} would be to show that some potential function or exponential moment is increasing slowly in expectation, as is done with several analyses of (sub)martingales. However, this mode of analysis encounters significant difficulties due to the fact that the partial sum $w_{i-1}$ and $v_i$ might be orthogonal. This prevents us from arguing that some potential function is pointwise a (sub)martingale.

We instead maintain a distributional guarantee that the distribution of $w_i$ over executions of \cref{algo:balance} is less ``spread" out than an associated normal distribution. This allows us to transfer the tail bounds on normal distributions to the distribution of $w_i$.

\subsection{Preliminaries and conventions.}
For a vector $v$, we let $v(i)$ denote the $i$-th coordinate of $v$. For positive semidefinite matrices $A, B$ we write $A \pe B$ if $B-A$ is positive semidefinite. For a positive semidefinite matrix $M \in \mb{R}^{m \times m}$ we define $\mc{N}(0,M)$ as the normal distribution with covariance $M$, i.e. $\mathbb{E}_{x\sim \mc{N}(0,M)}[x_ix_j]=M_{ij}$. For a subset $S \subseteq \mb{R}^k$, we write $1_S$ for the indicator function of the set $S$.

\subsection{Concurrent and Independent Work}
In independent and concurrent work, Bansal, Jiang, Meka, Singla, and Sinha \cite{BJMSS20} (building on the techniques of \cite{BJSS19}) achieve similar guarantees to the present work (with worse poly-log factors) for the online \komlos~problem restricted to the setting where vectors are sampled randomly from a fixed distribution $\frak{p}$ inside the unit sphere. However \cite{BJMSS20} uses potential based techniques as in \cite{BJSS19} and thus their results due not extend to minimizing discrepancy in the (more general) oblivious adversary model which is the primary focus of this paper. They also consider two extensions of this problem. In the first they consider a more general problem of splitting a set of incoming vectors (drawn from a stochastic distribution) into $k$ families such that the discrepancy between any two families is small. In the second they prove a more general result showing that one can balance vectors chosen from a known distribution $\frak{p}$ inside the unit sphere against an arbitrary norm induced by a symmetric convex body $K$ with Gaussian measure at least $1/2$. We believe our methods extend to these settings and plan on addressing this in future work. 

\section{Analysis}\label{sec:analysis}
\subsection{Properties of Spreading}
We now define the key notion for the analysis -- the notion of one random variable being a \emph{spread} of another.
\begin{definition}
We say that random variables $Y$ on $\mb{R}^n$ is a spread of random variable $X$ on $\mb{R}^n$ if there exists a coupling of $X$ and $Y$ such that $\mb{E}[Y|X] = X$.
\end{definition}
The univariate notion of the definition above appears in mathematical economics literature under the name ``mean-preserving spread'' \cite{RS70} and is closely related to ``second-order stochastic dominace'' \cite{HL69,HR69,RS70}. As defined, the name spread may seem unintuitive, but consider the coupling between $X$ and $Y$ such that $\mb{E}[Y|X] = X$. Then the random variable $Y-X$ conditional on $X$ has mean $0$. Thus, if $Y$ is a spread of $X$, one can obtain $Y$ by first sampling $X$, and then adding a mean $0$ random variable $Z$ whose distribution may depend on the specific value of $X$ sampled. Furthermore, since whether $Y$ is a spread of the random variable $X$ only depends on the respective distributions of $Y$ and $X$, we will often refer to distributions as spreads of one another. Equivalently, $X$ and $Y$ can be coupled so that $X,Y$ form a two-step martingale; the former perspective, however, is far more useful here. 
\begin{lemma}
\label{lemma:average}
Let distribution $Y$ be a spread of $X$. For any convex function $\Phi: \mb{R}^n \to \mb{R}$, we have
that $\mb{E}_{x \sim X}\Phi(x) \le \mb{E}_{y \sim Y} \Phi(y).$
\end{lemma}
\begin{proof}
Couple $X$ and $Y$ in the manner which demonstrates that $Y$ is a spread of $X$. Then
\begin{align*}
\mb{E}_{y\sim Y}\Phi(y) &= \mb{E}[\Phi(Y)|X] \ge \mb{E}\Phi(\mb{E}[Y|X]) = \mb{E}\Phi(X)
\end{align*}
where we have used Jensen's inequality and that $\mb{E}[Y|X] = X.$
\end{proof}
Spreading is transitive and preserved by linear transformations.
\begin{lemma}[Spreading is transitive] \label{lemma:transitive}
If $Z$ is a spread of $Y$ and $Y$ is a spread of $X$, then $Z$ is a spread of $X$.
\end{lemma}
\begin{lemma}[Linear transformations maintain spreading]
\label{lemma:linear}
If $Y$ is a spread of $X$, then for any linear transformation $M$ on $\mb{R}^n$ we have that $MY$ is a spread of $MX.$
\end{lemma}
The following is a slightly more abstract property of spreading that we need for \cref{thm:balance}.
\begin{lemma}\label{lemma:technical}
Consider random variables $X$, $Y$, $W$, and $Z$. Suppose that $W$ is a spread of $X$ and $Z$ is a fixed random variable such that $Z$ is a spread of the conditional distribution of $Y-X$ given $X =x$ for any value $x$. Then $W+Z$, where $W$ and $Z$ are sampled independently, is a spread of $Y$. 
\end{lemma}
\begin{remark}
It is implicit in the above definition that $X$ and $Y$ live on the same probability space.
\end{remark}
\begin{proof}
The proof produces the desired coupling between $Y$ and $W+Z$ as follows.
\begin{itemize}
    \item Sample $W$ and $X$ using the coupling between $W$ and $X$ which demonstrates that $W$ is a spread of $X$. 
    \item Then sample $Y$ from the conditional distribution of $Y$ given $X$ so that $W$ and $Y$ are conditionally independent given $X$.
    \item Finally sample $Z$ from the coupling of $Y-X$ and $Z$ (given the value of $X = x$) so that $(Y,Z)$ and $W$ are conditionally independent given $X$.
\end{itemize}
We claim that the marginal distribution of $W+Z$ is as if $W$ and $Z$ were sampled independently. This follows by noting that $W$ and $Z$ are conditionally independent given $X$ in this coupling and that the distribution of $Z$ conditional on $X$ is independent of the value of $X$ by hypothesis. Finally we prove that $W+Z$ is a spread of $Y$. Note that we have
\[\mb{E}[W+Z|X,Y] = \mb{E}[W|X,Y] + \mb{E}[Z|X,Y] = \mb{E}[W|X] + \mb{E}[Z|X,Y-X] = X + (Y-X) = Y\]
where each equality follows by construction. Therefore, $\mb{E}[W+Z|Y] = Y.$ The result follows. 
\end{proof}

Finally, any bounded variable with mean $0$ can be spread by the appropriate normal distribution.
\begin{lemma}[Spreading real variables by Gaussians]\label{lemma:2spread}
Let $X$ be a real-valued random variable with $\mb{E}[X] = 0$ and $|X|\le C$. Then $G=\mc{N}(0,\pi C^2/2)$ is a spread of $X$.
\end{lemma}
\begin{proof}
We first prove that the Bernoulli distribution with equal weight on $\pm C$ is a spread of any such variable $X$. To see this, define $Y$ conditional on $X$ to be $C$ with probability $(X + C)/(2C)$ and $-C$ with probability $(-X + C)/(2C)$. Then $Y$ is clearly supported on only $\pm C$ and \[\mb{E}[Y|X] = \frac{C(X+C)}{2C} + \frac{-C(-X+C)}{2C} = X,\] so the Bernoulli distribution with equal weight on $\pm C$ is a spread of $X$.
To see that $G$ is a spread of the Bernoulli distribution with equal weight on $\pm C$, let $Z$ be distributed as $|G|$ if $X = C$ and $-|G|$ otherwise. The result follows by noting that 
\[\mb{E}[|G|] = \sqrt{\frac{\pi}{2}}C\mb{E}_{Z\sim \mc{N}(0,1)}[|Z|] = C\]
and $Z$ is distributed as $G$.
\end{proof}
\subsection{Proof of \cref{thm:balance}}
We now formalize the notion of the distribution at time $i$ induced by \cref{algo:balance}. For $i \in [t]$, this is defined to be the distribution of $w_i$, except with all mass where the algorithm failed (line \ref{line:fail}) moved to $0.$
\begin{definition}[Distribution induced by \Balance]
We define the distribution induced by \\ $\Balance(v_1,\ldots,v_t,\delta)$ recursively.
\begin{itemize}
    \item At $i=0$ we have all mass of the distribution at $0$.
    \item Move all mass with $|\l w_i,v_{i+1} \r| > c$ or $\|w_i\|_\infty > c$ to $0$ -- this mass will stay at $0$ for the remainder of the process. We refer to such $w_i$ as being \emph{corrupted}.
    \item For the remaining mass $i \in [t]$, evolve the distribution according to lines \ref{line:prob}, \ref{line:move}.
\end{itemize}
$\mc{D}_i(\Balance(v_1,\ldots,v_t,\delta))$ will denote the distribution of the vector $w_i$ after $i$ time steps of the above procedure.
\end{definition}
One key observation is that at each stage we have that $\mc{D}_i(\Balance(v_1,\ldots,v_t,\delta))$ is symmetric about the origin. We will ultimately compare $\mc{D}_i(\Balance(v_1,\ldots,v_t,\delta))$ to $\mc{N}(0,M_i)$ for an appropriate set of covariance matrices $M_i$.
\begin{definition}
\label{def:covar}
Fix $L = 2\pi $ and $c\ge 1$. Define $M_0 = 0$. For $i\ge 1$ define $M_i$ inductively as 
\[M_i = (I - c^{-1}v_iv_i^{T})M_{i-1}(I - c^{-1}v_iv_i^{T}) + L v_iv_i^{T}.\]
\end{definition}
We now note that these covariance matrices are pointwise upper bounded independent of time.
\begin{lemma}
\label{lemma:covar}
Let $M \se 0$, $c \ge 1$, $L \ge 0$. If $M \pe LcI$ then for any vector $v \in \mb{R}^n$ with $\|v\|_2 \le 1$
\[ M' = (I-c^{-1}vv^T)M(I-c^{-1}vv^T) + Lvv^T \] satisfies $0 \pe M' \pe LcI.$
\end{lemma}
\begin{proof}
It is direct that $0 \pe M'$. Note that
\begin{align*} 
M' &\pe (I-c^{-1}vv^T)LcI(I-c^{-1}vv^T) + Lvv^T
\\ &= LcI - L(2-c^{-1}\|v\|_2^2)vv^T + Lvv^T \pe LcI
\end{align*}
as $2-c^{-1}\|v\|_2^2 \ge 2-c^{-1}\ge 1$.
\end{proof}
Applying \cref{lemma:covar} inductively gives us the following immediate corollary.
\begin{corollary}
\label{cor:covar}
For all $i \in [t]$ we have that $0\pe M_i \pe LcI$.
\end{corollary}
The following lemma is the key step in our analysis, where we show that the distribution $\mc{D}_i(\Balance(v_1, \cdots, v_t))$ is spread by normal distributions with covariance matrices defined in \cref{def:covar}. 
\begin{lemma}
\label{lemma:spreadmain}
$\mc{N}(0,M_i)$ is a spread of $\mc{D}_i(\Balance(v_1,\ldots,v_t,\delta))$ for all times $i \in [t]$.
\end{lemma}
\begin{proof}
For simplicity, we write $\mc{D}_i := \mc{D}_i(\Balance(v_1,\ldots,v_t,\delta))$ throughout the proof. We can compute that if the algorithm does not fail (line \ref{line:fail}) then
\[ \mb{E}[w_i | w_{i-1}] = w_{i-1} + (2p_i-1)v_i = w_{i-1} - c^{-1}v_iv_i^Tw_{i-1} = (I-c^{-1}v_iv_i^T)w_{i-1}. \]
Define $w_{i-1}'$ to be $w_{i-1}$ except when $w_{i-1}$ became corrupted -- in this case set $w_{i-1}'$ to $0$. Let $\mc{D}_{i-1}'$ be the distribution of $w_{i-1}'$. As $\mc{D}_{i-1}$ is symmetric, $\mc{D}_{i-1}$ is a spread of $\mc{D}_{i-1}'$. Define random variables
\[ R(w_{i-1},v_i) = \begin{cases}
0 & \text{ if } w_{i-1} \text{ is corrupted,} \\
1+c^{-1}\l w_{i-1},v_i \r & \text{ with probability } p_i \text{ if } w_{i-1} \text{ is not corrupted,} \\
-1+c^{-1}\l w_{i-1},v_i \r & \text{ with probability } 1 - p_i \text{ if } w_{i-1} \text{ is not corrupted,}
\end{cases}
\]
where $p_i$ is defined in line \ref{line:prob-init}. By definition, $w_i$ is distributed as
\[ (I-c^{-1}v_iv_i^T)w_{i-1}' + R(w_{i-1}',v_i)v_i. \]

By induction and \cref{lemma:linear} we have that \[ \mc{N}(0,(I-c^{-1}v_iv_i^T)M_{i-1}(I-c^{-1}v_iv_i^T)) \] is a spread of $(I-c^{-1}v_iv_i^T)\mc{D}_{i-1}$, which is a spread of $(I-c^{-1}v_iv_i^T)\mc{D}_{i-1}'$ by \cref{lemma:transitive}.

Note that by definition, $R(w_{i-1},v_i)$ is mean $0$ and supported on $[-2,2]$ as $|\l w_{i-1},v_i \r| \le c$ if $w_{i-1}$ is not corrupted. By \cref{lemma:2spread}, we have that $\mc{N}(0,2\pi)$ is a spread of $R(w_{i-1}',v_i)$ for each $w_{i-1}'$. Thus by \cref{lemma:technical}, we have that
\begin{align*}
&\mc{N}(0,(I-c^{-1}v_iv_i^T)M_{i-1}(I-c^{-1}v_iv_i^T)) + \mc{N}(0,2\pi v_iv_i^T) \\
&= \mc{N}(0,(I-c^{-1}v_iv_i^T)M_{i-1}(I-c^{-1}v_iv_i^T) + 2\pi v_iv_i^T) = \mc{N}(0, M_i)
\end{align*} is a spread of $\mc{D}_i$, as desired.
\end{proof}
We can get tail bounds on $M_i$ because they are always bounded (\cref{cor:covar}).
\begin{lemma}
\label{lemma:gauss-bound}
For any vector $u \in \mb{R}^n$ with $\|u\|_2 \le 1$ we have that 
\[\mb{E}_{x\sim \mc{N}(0,M_i)}\big[e^{\l  x,u\r^2/(4Lc)}\big] \le \sqrt{2}.\]
\end{lemma}
\begin{proof}
Note that $\l x,u\r$ is distributed as $\mc{N}(0,u^TM_iu)$ and $u^TM_iu \le Lc$ by \cref{cor:covar}. This implies that $\mc{N}(0,Lc)$ is a spread of $\mc{N}(0,u^TM_iu)$. The result then follows by noting that 
\[\mb{E}_{x\sim \mc{N}(0,u^TM_iu)}\big[e^{x^2/(4Lc)}] \le \mb{E}_{x\sim \mc{N}(0,Lc)}\big[e^{x^2/(4Lc)}] = \sqrt{2}, \]
where we have used \cref{lemma:average} on the convex function $e^{x^2/(4LC)}.$
\end{proof}
We are now ready to complete the proof of \cref{thm:balance} by combining the fact that $\mc{D}_i$ is spread by $\mc{N}(0,M_i)$ with the tail bounds in \cref{lemma:gauss-bound}.
\begin{proof}[Proof of \cref{thm:balance}]
It suffices to bound the total amount of additional corrupted mass in $\mc{D}_i$ compared to $\mc{D}_{i-1}$. To bound this note for any vector $u \in \mb{R}^n$ with $\norm{u}_2 \le 1$ that
\begin{align} &\mr{Pr}_{w \sim \mc{D}_i}\left[|\l w, u\r| > c\right] \le e^{-c^2/(4Lc)}\mb{E}_{w \sim \mc{D}_i}\left[e^{\l w, u\r^2/(4Lc)}\right] \label{eq:gauss1}
\\ &\le e^{-c/4L}\mb{E}_{w \sim \mc{N}(0,M_i)}\left[e^{\l w, u\r^2/(4Lc)}\right] \le \sqrt{2}e^{-c/4L} \le (2nt)^{-1}\delta \label{eq:gauss2}
\end{align}
by the choice of $c$. Here, we have used that $\mc{N}(0,M_i)$ is a spread of $\mc{D}_i$ through \cref{lemma:spreadmain}, \cref{lemma:average} on the convex function  $e^{\l x,u\r^2/(4LC)}$, and \cref{lemma:gauss-bound}. Now, at step $i \in [t]$ union bound over the choices $u = v_i, e_1, e_2, \ldots, e_n$ where $e_j$ denotes the unit basis vector for coordinate $j$.
\end{proof}

\begin{proof}[Proof of \cref{thm:komlos}]
Modify \cref{algo:balance} to use $c = 30\log(t/\delta)$ for $\delta = \frac{1}{\mr{poly}(t)}$ and note that we no longer need to maintain that $\snorm{w_i}_{\infty}\le c$ at every step and instead only at the end. As in \cref{eq:gauss1,eq:gauss2} we have that
\[ \Pr_{w_{i-1}\sim \mc{D}_{i-1}}[|w_{i-1},v_i| > c] \le \sqrt{2}e^{-c/4L} \le \delta/t. \]
Therefore, a union bound over $i \in [t]$ shows that the algorithm does not fail for steps $i \in [t]$ with probability $1-\delta.$ Again, as in \cref{eq:gauss1,eq:gauss2} for any basis vector $e_j$ we have that
\begin{align*}
&\Pr_{w_t \sim \mc{D}_t}[|w_t,e_j| > \sqrt{8cL\log n}] \le \sqrt{2}e^{-8cL \log n / 4cL} \le \frac{2}{n^2}.
\end{align*}
Now union bounding over all $j \in [n]$ using that $\sqrt{8cL \log n} = O(\sqrt{\log t \cdot \log n})$ gives the bound.
\end{proof}

\section{Applications}
\label{sec:applications}
We can obtain improved bounds for several geometric discrepancy problems given in \cite{BJSS19}. Additionally, given \cref{thm:balance} our algorithms are simpler, and do not require the Haar basis / wavelets used in \cite{BJSS19}.
\subsection{Interval discrepancy.}
The $d$-dimensional interval discrepancy problem is to assign weights $\eps_i \in \{-1,1\}$ to vectors $x_1, \cdots, x_t \in [0,1]^d$ to minimize the discrepancy
\[ \disc_I^i(k) := \left|\eps_11_I(x_1(k)) + \eps_21_I(x_2(k)) + \cdots + \eps_i1_I(x_i(k)) \right| \] over all intervals $I \subseteq [0,1]$ and $i\in[t], k \in [d].$

Applying \cref{thm:balance} gives bounds for the $d$-dimensional interval discrepancy problem matching the known lower bounds up to $\mr{poly}(\log dt)$ factors shown in \cite{BJSS19} Theorem 1.2. Additionally, our result works when $x$ is sampled from an arbitrary known distribution on $[0,1]^d$, instead of only in the case where the distribution is uniform.
\begin{theorem}
\label{thm:interval}
There is an online algorithm which for vectors $x_1, x_2, \dots, x_t$ chosen from a known distribution $\frak{p}$ on $[0,1]^d$ maintains $\disc_i^t(k) = O(\sqrt{d}\log^2(dt/\delta))$ for all intervals $I \subseteq [0,1]$ and $i\in[t], k \in [d]$ with probability $1-\delta.$
\end{theorem}
\begin{proof}[Sketch]
For simplicity we consider the case when the distribution $\frak{p}$ is absolutely continuous with respect to the Lebesgue measure; this assumption can be removed with some care. Define the quantiles \[ q_k^j := \inf \{ y : \mr{Pr}_{x \sim \frak{p}}[x(k) \le y] \ge j/(dt) \} \] for $k \in [d], 1 \le j \le dt.$ Let $0 = r_0 \le r_1 \le \dots \le r_{d^2t} = 1$ be a sorting of all the quantiles $q_k^j$ for $k \in [d], 1 \le j \le dt.$ By increasing $d,t$ by constants we can assume that $d^2t = 2^K$ for some integer $K$. Now, consider the set of intervals \[ J := \{ [r_{a \cdot 2^b}, r_{(a+1) \cdot 2^b}] : 0 \le a < 2^{K-b}, 0 \le b \le K \}. \]
Note that each point $x \in [0,1]$ is in $O(d \log dt)$ total such intervals, over all $d$ dimensions we wish to consider. Therefore, the proof of Theorem \ref{thm:balance} shows that the distribution of the discrepancy over the dyadic intervals is $O(\sqrt{d}\log(dt/\delta))$-subgaussian. Now, every interval can be written as a union of $O(\log dt)$ intervals in $J$, plus small error terms on the ends. Therefore, the corresponding vector has $\ell_2$ norm bounded by $O(\log^{1/2}(dt))$ Because there are $O(d^3t^2)$ total intervals, all have discrepancy \[ O(\sqrt{d}\log(dt/\delta)\log^{1/2}(dt)\log^{1/2}(d^4t^2/\delta)) = O(\sqrt{d}\log^2(dt/\delta)) \] with high probability as desired.
\end{proof}

\subsection{Online Tusn\'{a}dy's problem.}
Tusn\'{a}dy's problem is the following -- given points \\ $x_1, \cdots, x_t \in \mb{R}^d$ to minimize the maximum discrepancy over boxes $B \subseteq \mb{R}^n:$
\[ \disc(B) := \left|\eps_11_B(x_1) + \eps_21_B(x_2) + \cdots + \eps_t1_B(x_t)\right|. \] The best known upper bound is $O_d(\log^{d-1/2} t)$ \cite{Nikolov17} and the best known lower bound is $\Omega_d(\log^{d-1} t)$ \cite{MN15}. One can ask an analogous online version, which is to minimize the maximum over boxes $B$ and $k \in [t]$ of \[ \disc_i(B) := \left|\eps_11_B(x_1) + \eps_21_B(x_2) + \cdots + \eps_i1_B(x_i)\right|. \]

We can apply \cref{thm:balance} to the online Tusn\'{a}dy problem in the case where the vectors $x_1,\dots,x_t$ are sampled from a known distribution $\frak{p}$ on $[0,1]^n$. Our bounds are more general and improve over the $O_d(\log^{2d+1} t)$ bounds of \cite{BJSS19,Dwivedi19}, which only worked in the case of product distributions.
\begin{theorem}
\label{thm:tusnady}
There is an online algorithm which for vectors $x_1, x_2, \dots, x_t$ chosen from a known distribution $\frak{p}$ on $[0,1]^d$ maintains $\disc_i(B) = O_d(\log^{d+1}(dt/\delta))$ for all boxes $B \subseteq [0,1]^d$ and $i \in [t]$ with probability $1-\delta.$
\end{theorem}
\begin{proof}[Sketch]
As before, for simplicity, we only consider the case where $\frak{p}$ is absolutely continuous with respect to the Lebesgue measure. Compute the quantiles as done in the proof of \cref{thm:interval}. Now, build the dyadic decomposition as in \cref{thm:interval} one dimension at a time recursively to build dyadic boxes. Each point $x \in [0,1]^d$ is in $O_d(\log^d(dt))$ such dyadic boxes. Thus, the proof of \cref{thm:balance} shows that the distribution of discrepancy over the dyadic intervals is $O_d(\log^\frac{d+1}{2}(dt/\delta))$ subgaussian. 
Now, every box can be written as the union of $O_d(\log^d(dt))$ dyadic boxes, plus small errors. Therefore, the corresponding vecotr has $\ell_2$ norm at most $O_d(\log^{d/2}(dt))$. As there are $O_d((dt)^{2d})$ such boxes to consider, we achieve discrepancy at most
\[ O_d(\log^\frac{d+1}{2}(dt/\delta) \log^{d/2}(dt) \log^{1/2}(dt/\delta)) = O_d(\log^{d+1}(dt/\delta)) \] with high probability.
The rounding error can be handled with a Chernoff bound.
\end{proof}
We end by noting that our bounds for Theorem \ref{thm:interval} and \ref{thm:tusnady} work in the offline setting as well. Indeed, knowing the input points allows to form the dyadic partitions just as described in the proofs. 
For the interval disrepancy problem in the offline setting, an upper bound of $O(\sqrt{d}\log t)$ \cite{SST97} is known, as well as lower bounds of $\Omega(\log t)$ \cite{NNN12,Franks18} and $\Omega(\sqrt{d})$. Therefore, our bound for the interval discrepancy problem is only a single $O(\log(dt))$ factor off the best known offline bound, and a $O(\log^2(dt))$ factor off the known lower bound.
Additionally, our bound is only a $\log^{3/2}(dt)$ factor off the best known \emph{offline} bound for Tusn\'{a}dy's problem, and a $\log^2(dt)$ factor off the known lower bound.
\section*{Acknowledgements}
The authors thank Noga Alon, Vishesh Jain, Mark Sellke, Aaron Sidford, and Yufei Zhao for helpful comments regarding the manuscript. We thank Nihkil Bansal for observations which led to improvements of several logarithmic factors in \cref{thm:interval} and \cref{thm:tusnady}. R.A. is supported by an NSF Graduate Research Fellowship. Y.L. was supported by the Department of Defense (DoD) through the National Defense Science and Engineering Graduate Fellowship (NDSEG) Program.
\bibliographystyle{amsplain0.bst}
\bibliography{main.bib}

\end{document}